\documentclass[journal,10pt]{IEEEtran}
\usepackage{amsmath,amsfonts}
\usepackage{amsmath}
\usepackage{array}
\usepackage[font=small]{caption}
\usepackage{textcomp}
\usepackage{stfloats}
\usepackage{url}
\usepackage{verbatim}
\usepackage{graphicx}
\usepackage{cite}
\usepackage{tikz}
\usepackage{amsthm}
\usepackage{subcaption}
\usepackage[normalem]{ulem}
\usepackage{etoolbox}
\usepackage{amssymb}

\usepackage{mathtools}
\newtheorem{proposition}{Proposition}
\usepackage{breqn}
\usepackage[multiple]{footmisc}
\newcommand{\widesim}[2][1.5]{
	\mathrel{\overset{#2}{\scalebox{#1}[1]{$\sim$}}}
}

\definecolor{ForestGreen}{RGB}{34,139,34}
\usepackage{bbm}
\newtheorem{theorem}{Theorem}

\allowdisplaybreaks

\usepackage{algorithm}
\usepackage{algpseudocode}
\usepackage{setspace}
\usepackage{caption}
\usepackage[font=small,skip=0.3pt]{caption}
\usepackage{subcaption}
\begin{document}	
	\title{Does an IRS Degrade Out-of-Band Performance?}
	\author{L. Yashvanth,~\IEEEmembership{Student Member,~IEEE}, and Chandra R. Murthy,~\IEEEmembership{Fellow,~IEEE}
	\thanks{The authors are with the Dept. of ECE, Indian Institute of Science, Bengaluru, India 560 012. (E-mails: \{yashvanthl, cmurthy\}@iisc.ac.in).}
	\thanks{\textcolor{black}{This work was supported in parts by the Qualcomm Innovation Fellowship, a research grant from MeitY, Govt. of India, and the Prime Minister’s Research Fellowship, Govt. of India.}}
	}
\maketitle
\begin{abstract}
Intelligent reflecting surfaces (IRSs) were introduced to enhance the performance of wireless systems. However, from a cellular service provider’s view, a concern with the use of an IRS is its effect on out-of-band (OOB) quality of service. Specifically, given two operators, say X and Y, providing services in a geographical area using non-overlapping frequency bands, if operator-X uses an IRS to optimally enhance the throughput of its users, does the IRS degrade the performance of operator-Y? We answer this by deriving the ergodic sum spectral efficiency (SE) of both operators under round-robin scheduling. We also derive the complementary cumulative distribution function of the change in effective channel at an OOB user with and without the IRS, which provides deeper insights into OOB performance. Surprisingly, we find that even though the IRS is randomly configured from operator-Y’s view, the OOB operator still benefits from the IRS, witnessing a performance enhancement for free. This happens because the IRS introduces additional paths between the nodes, increasing the signal power at the receiver and providing diversity benefits. We verify our findings numerically and conclude that an IRS is beneficial to every operator, even when the IRS is deployed to optimally serve only one operator.
 \end{abstract}
\section{Introduction}
Intelligent reflecting surfaces (IRS) have been extensively studied in the literature as a means to enhance the performance of both indoor and outdoor wireless systems~\cite{Basar_IA_2019,RuiZhang_IRSSurvey_TCOM_2021}. An IRS is a passive electro-magnetic surface which comprises of IRS elements made of meta-materials. The IRS elements can introduce a small delay/phase shift in the radio frequency (RF) signal impinging on it before reflecting it, thereby allowing the IRS as a whole to steer the signal in any desired direction. This can be achieved by appropriately configuring the reflection coefficient at every IRS element. Further, due to the passive nature of IRS, the performance of an IRS aided system is enhanced while maintaining the energy efficiency~\cite{Huang_TWC_2019}. In order to obtain the professed benefits of an IRS, every IRS element needs to be configured to introduce a phase shift in the signal that is optimal to the scheduled user equipment (UE) in terms of a metric of interest, e.g., the signal-to-noise ratio (SNR). This, however, requires the knowledge of the channel state information (CSI) of all the links from the base station (BS) to the user through every IRS element~\cite{Basar_IA_2019,RuiZhang_IRSSurvey_TCOM_2021}.\\
\indent In practice, multiple network operators co-exist in a given geographical area, each operating in different frequency band. As a consequence, at a given point in time, multiple UEs are served by different operators in the system. In such a scenario, if an IRS is optimized to cater to the needs of one of the  operators, it is not clear whether the IRS will boost or degrade the performance of the other operators in the system. In particular, since the IRS elements are passive, they will reflect the RF signals impinging on them in all frequency bands. So, it is important to understand how an IRS which is controlled by only one operator affects the performance of other operators (called as \emph{out-of-band operator} in this paper).  Although a very few works consider the scenario of the presence of an IRS in multi-band systems~\cite{Cai_2022_TCOM,Cai_2020_CL_practicalIRS}, these works proceed along the lines of jointly optimizing the IRS phase configurations among all the operators. This approach requires inter operator co-ordination, which is not practical. Moreover, the solutions and analysis provided in these works are not scalable with number of operators (or frequency bands) in the system. More fundamentally, none of these works address the question of the out-of-band (OOB) performance even in the scenario of two operators operating in non-overlapping bands and the IRS is optimized for only one operator. In this paper, we address this question, and to the best of our knowledge, this is the first work which considers the effect of OOB performance due to the presence of an IRS under practical cellular network deployment scenarios.\\ 
\indent We consider a system with two network operators providing service in non-overlapping  frequency bands. We  analyze the OOB throughput performance in the presence of an IRS that is optimized to serve the users subscribed to an operator offering wireless services in a different frequency band. Specifically,
\begin{itemize}
\vspace{-0.05cm}
	\item  We derive the ergodic sum spectral efficiencies (SE) of the two operators as a function of the number of IRS elements, under round-robin  scheduling of UEs. We show that the ergodic sum-SE scales quadratically and linearly with the number of IRS elements for the in-band and OOB networks, respectively, even when the OOB operator has no control over the IRS in the environment.
	\item We provide an exact characterization of the complementary cumulative distribution function (CCDF) of the difference in the channel gain at an OOB UE with and without the IRS. We determine the probability with which the difference is non-negative as a function of the number of IRS elements, and show that the channel gain with an IRS \emph{stochastically dominates} the gain without the IRS. Further, the difference in the channel gains with and without the IRS is an increasing function of the number of IRS elements. This confirms that even an OOB UE witnesses benefits that monotonically increase with the number of IRS elements.
	\item Through numerical simulations, we illustrate that an IRS \emph{enhances} the OOB performance, i.e., the presence of an IRS is beneficial even if its reflection coefficients are chosen randomly from the OOB operator's viewpoint. 
\end{itemize}  

Our results show that deploying an IRS not only improves the throughput of the operator who controls the IRS phase configuration to optimally serve its own users, but also enhances the throughput of users associated with an OOB operator who has no control over the IRS, albeit by a smaller amount compared to the in-band users. Furthermore, in rich scattering environments, the throughput enhancement increases with the number of IRS elements deployed. Thus, deploying an IRS enriches the wireless channels, and thereby benefits all wireless operators in the area.
\section{System Model}\label{sec:ch_model}
We consider two mobile network operators X and Y who provide service to $K$ and $Q$ UEs, respectively. The UEs are arbitrarily distributed over a single cell covering the same geographical area, and operators X and Y use non-overlapping frequency bands. 
The base stations (BSs) of operators X and Y (referred to as BS-X and BS-Y, respectively) and the UEs are equipped with a single antenna, and all the channels in the systems undergo frequency-flat fading.\footnote{Extension to  general cases with multiple antennas and frequency selective channels does not change the main message, and is left to future work.} An $N$-element IRS is deployed by operator X in order to enhance the quality of service (QoS) to the UEs being served by it. That is, operator X configures the IRS with the optimal phase configuration for a UE scheduled by BS-X in every time slot. 
We model the downlink signal received at the $k$th UE (served by BS-X) as
\begin{equation}\label{eq:airtel_downlink}
\vspace{-0.05cm}
	y_{k} = \left(h_{d,k}+\mathbf{g}_{k}^T\boldsymbol{\Theta}\mathbf{f}^X\right)x_k + n_k,
\vspace{-0.05cm}
\end{equation} where $\mathbf{g}_k \in \mathbb{C}^{N \times 1}$ represents the channel from the IRS to the $k$th UE, $\mathbf{f}^X \in \mathbb{C}^{N \times 1}$ represents the channel from the BS-X to the IRS,  $\boldsymbol{\Theta} \in \mathbb{C}^{N \times N}$ is a diagonal matrix containing the IRS reflection coefficients of the form $e^{j\theta}$, and $h_{d,k}$ is the direct (non-IRS) path from the BS to the UE-$k$. Also, $x_k$ is the data symbol for UE-$k$ with average power $\mathbb{E}[|x_k|^2] = P$, and $n_k$ is the AWGN $\sim \mathcal{CN}(0,\sigma^2)$ at UE-$k$. Similarly, at UE-$q$ served by BS-Y, we have
\vspace{-0.1cm}
\begin{equation}\label{eq:jio_downlink}
\vspace{-0.1cm}
	y_{q} = \left(h_{d,q}+\mathbf{g}_{q}^T\boldsymbol{\Theta}\mathbf{f}^Y\right)x_q	+ n_q,
\end{equation} where the symbols have similar meanings as in~\eqref{eq:airtel_downlink}. Fig.~\ref{fig:Network_scenario} summarizes the considered network.
\begin{figure}[t!]
	\centering
	\includegraphics[width=\linewidth]{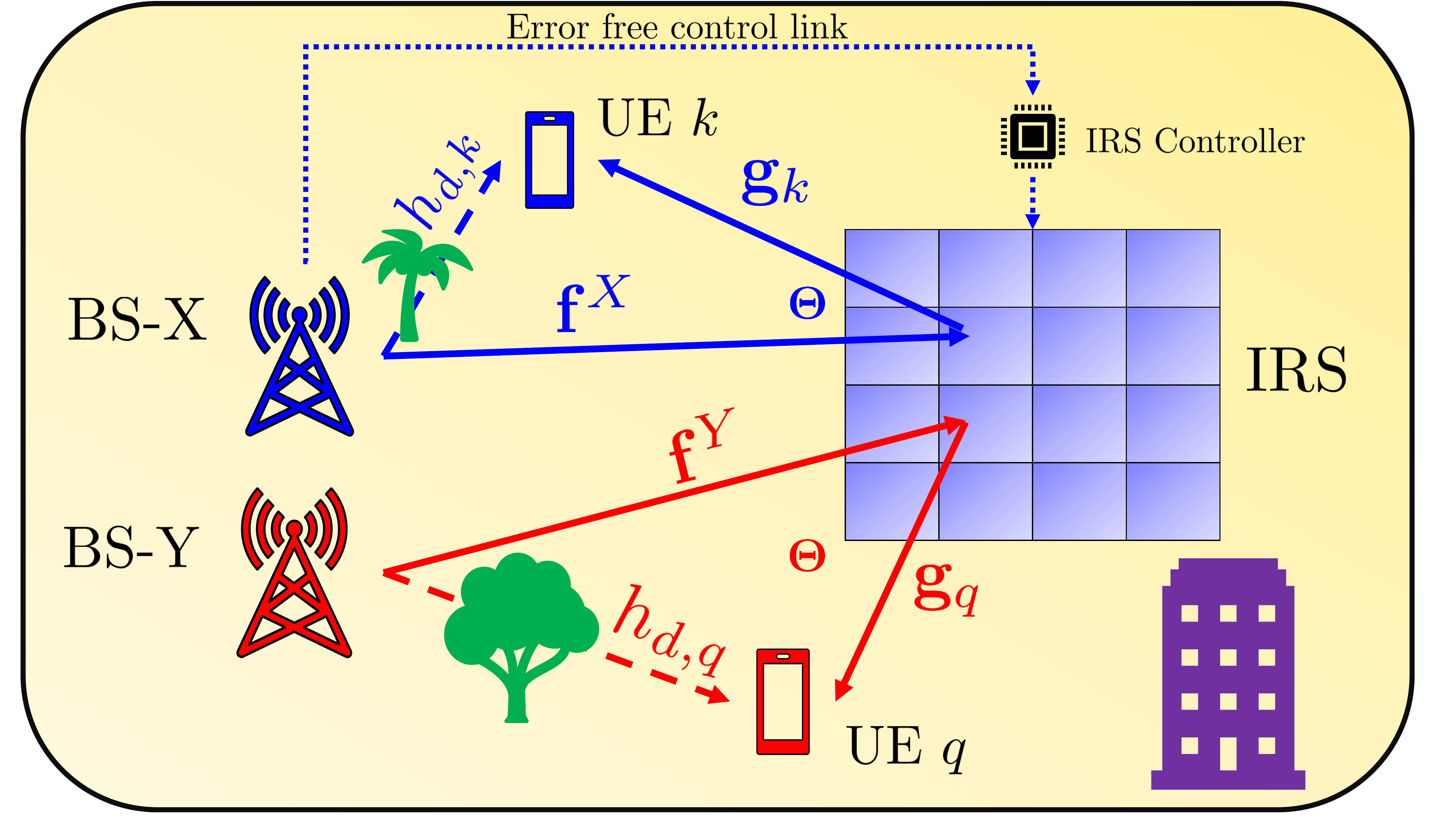}
	\caption{Network scenario of an IRS aided multiple-operator system.}
	\label{fig:Network_scenario}
	\vspace{-0.5cm}
\end{figure}
Similar to~\cite{Basar_TCOM_2022_iid}, we consider that all the fading channels in the system are statistically independent and follow the Rayleigh distribution.\footnote{\textcolor{black}{We consider that the IRS elements are placed sufficiently far apart so that the effect of channel spatial correlation is negligible.}} Specifically,	$h_{d,k} = \sqrt{\beta_{d,k}} \tilde{h}_{d,k}, h_{d,q} = \sqrt{\beta_{d,q}} \tilde{h}_{d,q}; \tilde{h}_{d,k}, \tilde{h}_{d,q}  \widesim[2]{\text{i.i.d.}} \mathcal{CN}(0,1)$; $\mathbf{g}_k = \sqrt{\beta_{\mathbf{g},k}} \tilde{\mathbf{g}}_k, \mathbf{g}_q = \sqrt{\beta_{\mathbf{g},q}} \tilde{\mathbf{g}}_q; \tilde{\mathbf{g}}_k,\tilde{\mathbf{g}}_q \widesim[2]{\text{i.i.d.}} \mathcal{CN}(\mathbf{0}, \mathbf{I}_N)$ ; $\mathbf{f}^X = \sqrt{\beta_{\mathbf{f}^X}} \tilde{\mathbf{f}}^X, \mathbf{f}^Y = \sqrt{\beta_{\mathbf{f}^Y}} \tilde{\mathbf{f}}^Y; \tilde{\mathbf{f}}^X,\tilde{\mathbf{f}}^Y \widesim[2]{\text{}} \mathcal{CN}(\mathbf{0}, \mathbf{I}_N)$. All terms of the form $\beta_x$ represent the pathloss factor in $x$th link.
\vspace{-0.2cm}
\section{Out-of-band Performance Analysis} \label{sec:OOB_perf_analysis}
As mentioned earlier, in this work, we consider a scenario where operator X deploys and controls an IRS in order to enhance the throughput of the users being served by it, and are interested in the effect of the IRS on an operator Y that is providing services in a different frequency band. Thus, in order to serve the $k$th UE, operator X configures the IRS with the rate-optimal phase angles~\cite{Basar_IA_2019,RuiZhang_IRSSurvey_TCOM_2021},\cite[Lemma~1]{Yashvanth_TSP_2023}
\begin{equation}
\theta ^{*}_{n,k} =\angle h_{d,k}- \left(\angle {f}^X_{n}+\angle{g}_{k,n}\right),\hspace{0.5cm} n=1, {\dots }, N, \label{eq:basic_optimal_angle_X}  
\end{equation}
which results in coherent addition of the signals along the direct path as well through all the IRS elements, leading to the maximum possible received SNR and hence data rate. The rate achieved by the $k$th UE is given by
\begin{equation} 
R_{k}^{BF} = \log _{2}\!\left( \! 1 + \frac {P}{\sigma ^{2}} \left| |h_{d,k}| + \sum_{n=1}^{N} |f^X_{n}{g_{k,n}}| \right|^{2}\right).
\vspace{-0.1cm}
\end{equation}
Now, due to the independence of the channels of the users served by operators X and Y, the phase configuration used by operator X to serve its own users appears as a \emph{random} phase configuration of the IRS for any UE served by operator~Y. In the sequel, we quantify the impact of the IRS on the throughput achieved by the users served by operator~Y which has no control over the phase configuration used at the IRS. 

In order to study the impact on the OOB performance, we consider the scheduling of UEs in a round-robin (RR) fashion at both BS-X and BS-Y. We note that the performance under opportunistic scheduling at either or both BSs can also be derived along similar lines, e.g., following the approach in~\cite{Yashvanth_TSP_2023}. Since the BSs are equipped with a single antenna, only one UE from each network is scheduled for data transmission in every time-slot. We characterize the OOB performance of the network by deriving the ergodic sum-SE of both the networks, and then infer the degree of degradation/enhancement of the OOB performance caused by the IRS to operator Y. The ergodic SE at UE-$k$ is 
\vspace{-0.1cm}
\begin{equation}\label{eq:basic_airtel_rate}
	\langle R_k^{(X)} \rangle = \mathbb{E}\!\left[\log_{2}\left(\!1 \!+\! \left|\sum\nolimits_{n=1}^N  |f^X_{n}{g_{k,n}}| 
		+   |h_{d,k}| \right|^2 \!\!\frac{P}{\sigma^2} \right)\!\right],
		\vspace{-0.1cm}
\end{equation} since the IRS is configured with the optimal phase configuration for (scheduled) UE-$k$.
On the other hand, the ergodic SE for (scheduled) UE-$q$ of operator Y is given by
\vspace{-0.1cm}
\begin{equation}\label{eq:basic_jio_rate}
	\langle R_q^{(Y)} \rangle = \mathbb{E}\!\left[\log_{2}\left(\!1\! +\! \left|\sum\nolimits_{n=1}^N  f^Y_{n}{g_{q,n}} 
		+   h_{d,q} \right|^2 \frac{P}{\sigma^2} \right)\right],
		\vspace{-0.1cm}
\end{equation} where we used the fact that the channels are circularly symmetric random variables, so  $f^Y_{n}g_{q,n}e^{j\theta_{n,k}} \stackrel{d}{=} f^Y_{n}g_{q,n}$ for any $\theta_{n,k}$. The expectations are taken with respect to the distribution of the channels to the respective UEs. With RR scheduling, the ergodic sum-SEs of two operators are given by
\vspace{-0.1cm}
\begin{align}\label{eq:sum-rate-template}
\vspace{-0.1cm}
	\bar{R}^{(X)}  \triangleq \frac{1}{K} \sum_{k=1}^K \langle R_k^{(X)} \rangle, \text{  and  }
	\bar{R}^{(Y)} \triangleq \frac{1}{Q}\sum_{q=1}^Q \langle R_q^{(Y)} \rangle.
	\vspace{-0.2cm}
\end{align}
Closed-form expressions for the ergodic sum-SEs are difficult to obtain due to the complicated distributions of the SNR terms in~\eqref{eq:basic_airtel_rate} and~\eqref{eq:basic_jio_rate}. \textcolor{black}{To explicitly capture the dependence of SE on the system parameters, we obtain an approximation by upper bounding the SE via Jensen's inequality. We show numerically in Fig.~\ref{fig:SE_txt_SNR} that the resulting bound is tight.} 
\vspace{-0.1cm}
\begin{theorem}\label{thm:rate_characterization}
	With RR scheduling, and under independent Rayleigh fading channels, the ergodic sum-SEs of the operators X and Y when the IRS is optimized to serve the UEs of operator X scale as
	\begin{multline}\label{eq:rate_airtel_rr}
	\vspace{-0.2cm}
		\bar{R}^{(X)}\ \textcolor{black}{\approx}  \  \frac{1}{K}\sum\nolimits_{k=1}^K\!  \log_{2}\left(1 + \! \left[N^2\left(\frac{\pi^2}{16} \beta_{r,k}\!\right)\!  \right. \right. \\ \left. \left. +  N\!\left(\!\beta_{r,k}\!-\!\frac{\pi^2}{16}\beta_{r,k}+\!\frac{\pi^{3/2}}{4}\sqrt{\beta_{d,k}\beta_{r,k}}\!\right)\! \!+\!\beta_{d,k}\right]\!\frac{P}{\sigma^2}\right)\!,
		\vspace{-0.5cm}
	\end{multline} where $\beta_{r,k} \triangleq \beta_{\mathbf{f}^X}\beta_{\mathbf{g},k}$, and
\vspace{-0.1cm}
	\begin{equation}\label{eq:rate_jio_rr}
	\vspace{-0.1cm}
		\bar{R}^{(Y)} \ \textcolor{black}{\approx} \ \frac{1}{Q}\sum\nolimits_{q=1}^Q  \log_{2}\left(1 + \left[N\beta_{r,q} + \beta_{d,q}\right]\frac{P}{\sigma^2} \right) ,
		\vspace{-0.15cm}
	\end{equation} where $\beta_{r,q} \triangleq \beta_{\mathbf{f}^Y}\beta_{\mathbf{g},q}$.
\end{theorem}
\vspace{-0.35cm}
\begin{proof}
	See Appendix~\ref{sec:app-1}.
\end{proof}
\vspace{-0.15cm}
We thus infer the following on the performance of an IRS aided system where several network operators co-exist in the same geographical area, serving  in different frequency bands:
	\begin{itemize}
		\item The IRS enhances the average received SNR by a factor of $N^2$ at any scheduled (in-band) UE of operator X when the IRS is optimized by BS-X. This is the benefit that operator X obtains by using an optimized $N$-element IRS.
		\item Operator Y, who does not control the IRS, also witnesses an enhancement of average SNR by a factor of $N$ free of cost, i.e., without any co-ordination with the IRS. This happens because the IRS makes the wireless environment more rich-scattering, and hence facilitates reception of multiple copies of the signals at the (out-of-band) UEs.
\end{itemize}
We now make the analysis more concrete by analyzing the stochastic behavior of the channels gains witnessed by an OOB UE with/without an IRS. Define the following random variables for an arbitrary OOB UE, say UE-$q$, served by BS-Y.
\vspace{-0.15cm}
\begin{equation}\label{eq:aux_RV_defn}
\vspace{-0.15cm}
\left|h_{1,q}\right|^2\triangleq \left|\sum_{n=1}^N  f^Y_{n}{g_{q,n}} 	+   h_{d,q} \right|^2;  \left|h_{2,q}\right|^2 \triangleq \left| h_{d,q} \right|^2.
\vspace{-0.1cm}
\end{equation} 
Note that $\left|h_{1,q}\right|^2$ and $\left|h_{2,q}\right|^2$ represent the channel power gain of UE-$q$ in the presence and absence  of the IRS, respectively. We  now characterize the change in the channel gain at UE-$q$ served by BS-Y in the presence and absence of the IRS as
\vspace{-0.1cm}
\begin{equation}\label{eq:snr_offset}
\vspace{-0.1cm}
	Z^{(Y)}_N \triangleq \left|h_{1,q} \right|^2 + \left(-{\mathbbm{1}_{\{N \neq 0\}}}\right) \left|h_{2,q} \right|^2.
\vspace{-0.1cm}
\end{equation} 
The random variable $Z^{(Y)}_N$ as defined above, provides a more conservative comparison of the instantaneous channel gains in the presence and absence of an $N$-element IRS over the entire support of their respective probability distributions. In fact, the event characterized by $Z^{(Y)}_N<0$ indicates the adverse effect on an OOB UE by the IRS, and in the sequel we prove that almost surely, $Z^{(Y)}_N$ becomes a non-negative random variable. Towards that end, we derive the CCDF of $Z^{(Y)}_N$ given by
\vspace{-0.1cm}
\begin{equation}\label{eq:ccdf_notation}
\vspace{-0.1cm}
	\text{CCDF}(Z^{(Y)}_N) \triangleq \bar{F}_{Z^{(Y)}_N}(z) = { \sf{Pr}}(Z^{(Y)}_N \geq z).
	\vspace{-0.1cm}
\end{equation} 
\begin{theorem}\label{thm:exact_ccdf}
The CCDF of the random variable $Z^{(Y)}_N$ when $N \textcolor{black}{(N>0)}$ is reasonably large\footnote{\textcolor{black}{In Sec.~\ref{sec:numerical_results}, we numerically show that the result in the theorem holds true even when $N$ is as small as $4$.}} is given by 
\vspace{-0.1cm}
\begin{equation}
\bar{F}_{Z^{(Y)}_N}(z)  = \left\{
\begin{array}{lr}
	1-\dfrac{1}{N\tilde{\beta}+2}\times e^{\left(\dfrac{z}{\beta_{d,q}}\right)}, \hspace{0.1cm} \text{ if } z < 0,\\
	\vspace{-0.4cm}
	\left(\dfrac{N\tilde{\beta}+1}{N\tilde{\beta}+2}\right)\times e^{-\left(\dfrac{z}{\beta_{d,q}\left(1+N\tilde{\beta}\right)}\right)}, \\
	\hfill \text{ if } z\geq 0.
\end{array}
\right.
\vspace{-0.3cm}
\end{equation} where $\tilde{\beta} \triangleq \dfrac{\beta_{r,q}}{\beta_{d,q}}$.
\end{theorem}
\begin{proof}
	See Appendix~\ref{sec:ccdf_proof}.
\end{proof}
\vspace{-0.2cm}
From the above theorem, we have 
$\bar{F}_{Z^{(Y)}_N}(0) = 1 - {1}/\left({2+N\tilde{\beta}}\right),$
i.e., for a given $\tilde{\beta}$, 
the probability that the SNR/gain offset in~\eqref{eq:snr_offset} is negative value decays as $\mathcal{O}\left({1}/{N}\right)$. 
Moreover, we see that $\bar{F}_{Z^{(Y)}_{N'}}(z) \geq \bar{F}_{Z^{(Y)}_{N''}}(z)$   $\forall z$ when $N' > N''$. 
Thus, we have the following proposition.
\begin{proposition}\label{prop:stochastic_dominance_sub_6}
For any $M,N \in \mathbb{N}\cup\{0\}$ with $M>N$, the random variable $Z^{(Y)}_{M}$ \emph{stochastically dominates}\footnote{A real-valued random variable $X$ is stochastically larger than, or stochastically dominates,  the random variable $Y$, written $X >_{st} Y$, if
	\vspace{-0.1cm}
	\begin{equation}
	\vspace{-0.1cm}
		{\sf {Pr}}(X>a) \geq {\sf {Pr}}(Y>a), \hspace{0.2cm} \text{for all }a.
\vspace{-0.05cm}
\end{equation} Note that, if the random variables $X$ and $Y$ have CCDFs $\bar{F}$ and $\bar{G}$, respectively, then $X >_{st} Y \Longleftrightarrow \bar{F}(a) \geq \bar{G}(a)$ for all $a$~\cite{Ross_2014_ProbabiltyBook}.} $Z^{(Y)}_{N}$. In particular, the OOB channel gain in the presence of the IRS stochastically dominates the channel gain in its absence.
\end{proposition}
\vspace{-0.15cm}
This proposition states that the random variables $\left\{Z^{(Y)}_{n}\right\}_{n\in \mathbb{N}\cup\{0\}}$ form a  sequence of \emph{stochastically larger} random variables as a function of the number of IRS elements, where $\mathbb{N}$ is the set of natural numbers. Thus,  the SNR offset increases with the number of IRS elements even at an OOB UE, i.e., the IRS  almost surely enhances the channel quality at an OOB UE at any point in time. Therefore, the performance of OOB operators \emph{does not degrade} even when the operator is completely oblivious to the presence of the IRS. In fact, this holds true for any number of operator in the area, and hence no operator will be at a disadvantage due to the presence of an IRS being controlled by only one operator. In the next section, we numerically illustrate these points.
\vspace{-0.2cm}
\section{Numerical Results}\label{sec:numerical_results}
\vspace{-0.1cm}
In this section, we validate the analytical results derived and empirically show that an IRS does not cause any degradation in the OOB performance. The BS-X and BS-Y are located at coordinates $(0,200)$, and $(200,0)$ (in metres), the IRS is at $(0,0)$ and UEs are located uniformly at random locations in the rectangular region with diagonally opposite corners $(0,0)$ and  $(200,200)$. \textcolor{black}{The path loss in each link is modeled as $\beta = C_0\left(d_0/d\right)^\alpha$ where $C_0$ is the path loss at the reference distance $d_0$, $d$ is the distance of the link, and $\alpha$ is the path loss exponent. We let $d_0=1$ metre; then at a carrier frequency of $750$ MHz, $C_0=-30$ dB}. We use $\alpha=1.5,2$ and $3$ in the BS X/Y-IRS, IRS-UE and BS X/Y-UE (direct) links, respectively,~similar to~\cite{Yashvanth_TSP_2023}.
The fading channels are randomly generated as per Sec.~\ref{sec:ch_model}.
\begin{figure}[t]
	\centering
	\includegraphics[width=\linewidth]{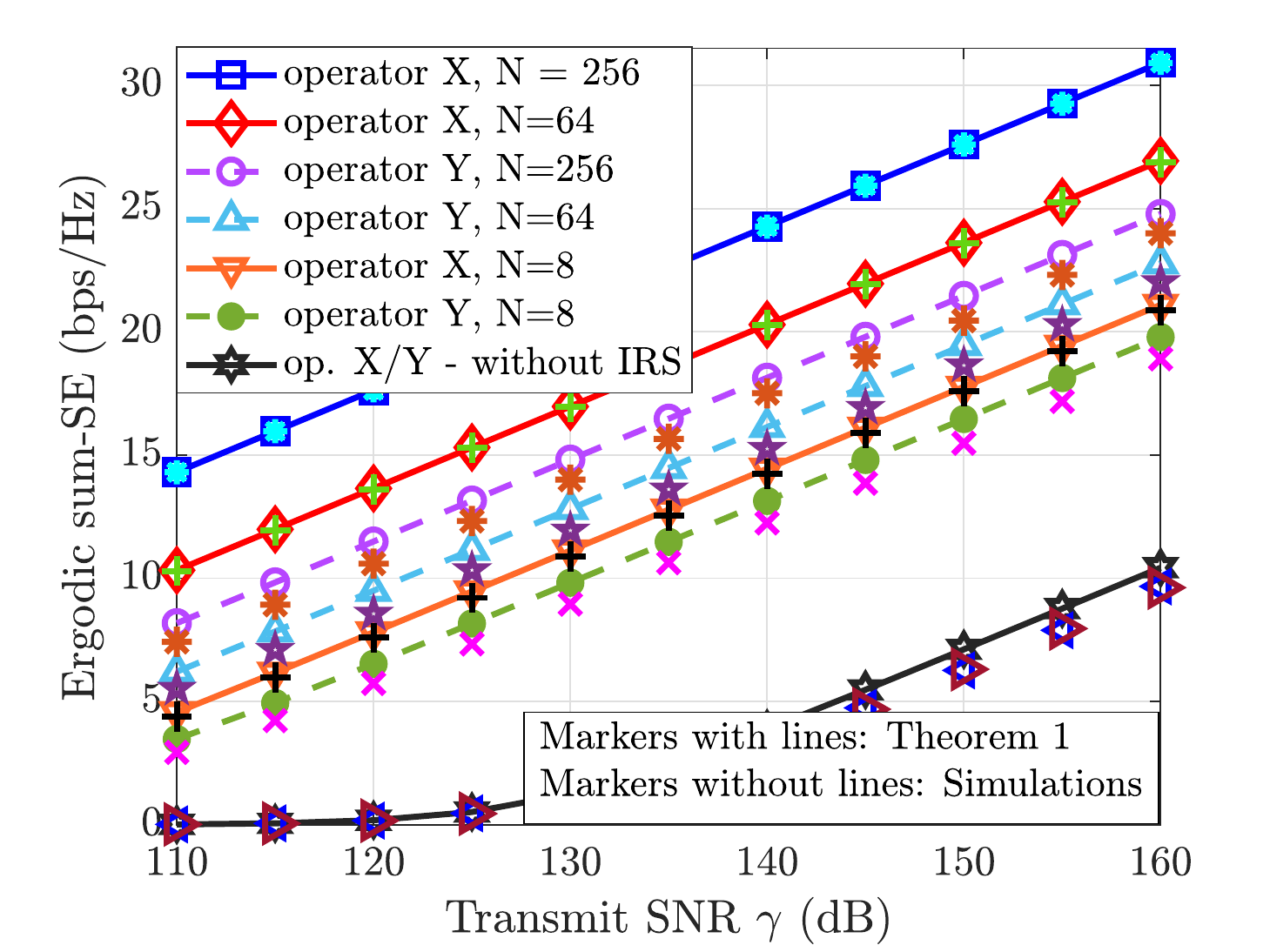}
	\caption{Spectral efficiency vs. transmit SNR.}
	\label{fig:SE_txt_SNR}
	\vspace{-0.55cm}
\end{figure}
\begin{figure}[t]
	\centering
	\includegraphics[width=\linewidth]{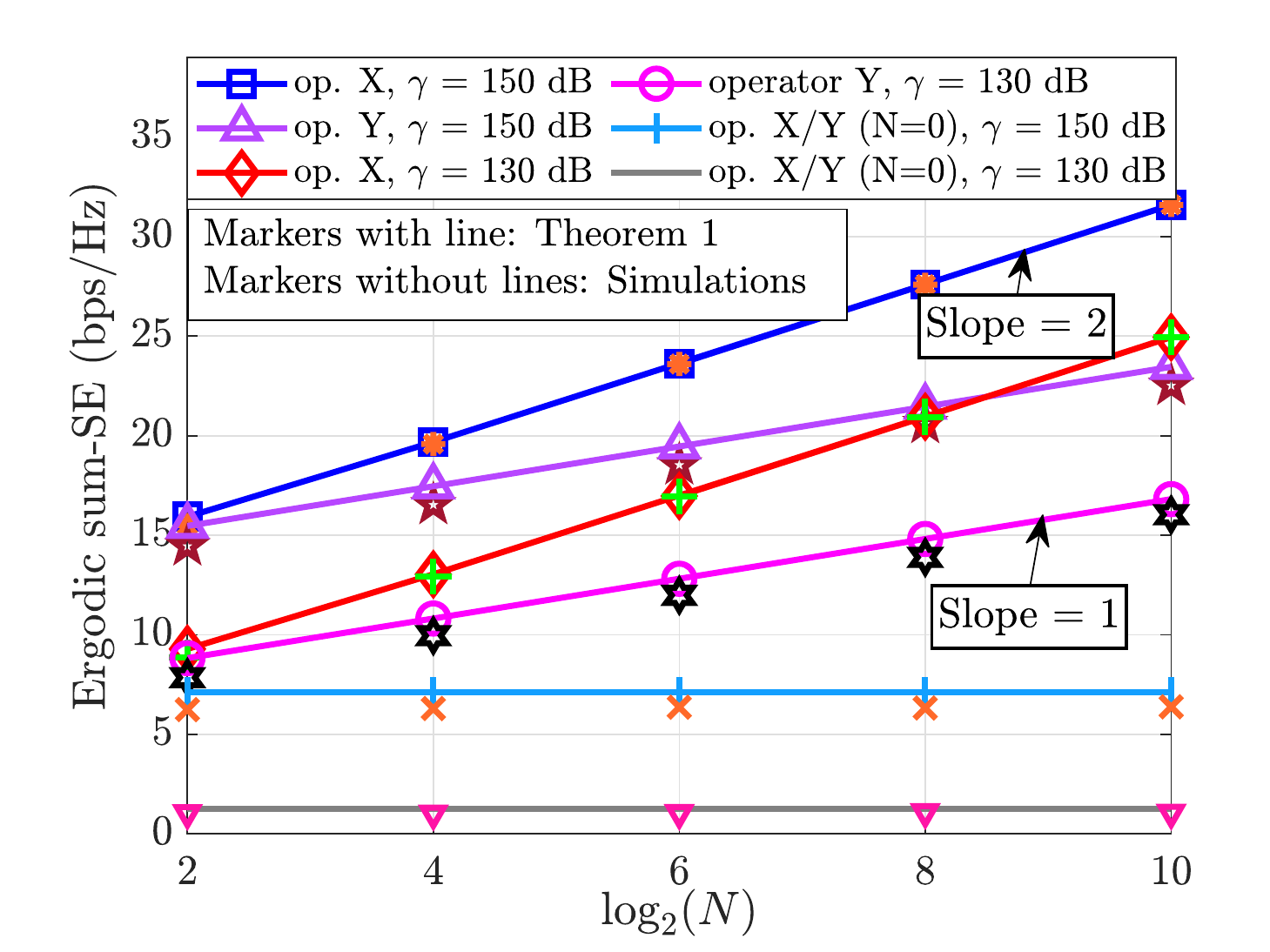}
	\caption{Spectral efficiency vs. $\log_2 (N)$.}
	\label{fig:SE_logN}
	\vspace{-0.7cm}
\end{figure}
In Fig.~\ref{fig:SE_txt_SNR}, we plot the empirical ergodic sum-SE vs. the \emph{transmit} SNR $\gamma \left(\triangleq P/\sigma^2\right)$ for both the operators as a function of the number of IRS elements.\footnote{We use the range of transmit SNR in $110-160$ dB. For e.g., consider $N=64$ and a transmit SNR of $135$ dB. This yields receive SNR $\approx$ $16$, $10$ dB with and without an IRS, respectively, at an OOB UE located at $(100, 100)$.} We also plot the sum-SE obtained from analytical expressions in Theorem~\ref{thm:rate_characterization}. We consider a scenario with $K=Q=10$, and all UEs being served over a total of $1000$ time slots under RR scheduling by the respective operators. Further, the IRS is optimized to serve the UEs of operator X. We see that while IRS uniformly enhances the signal strength for operator X at all SNRs, it also boosts the SNR for any UE served by BS-Y (for any number of IRS elements) which has no control over the IRS phase configurations. This corroborates our observation from Theorem~\ref{thm:rate_characterization} that the IRS does not degrade the OOB performance. Also, the derived analytical expressions tightly match with the simulated values, i.e., the approximation error due to the use of Jensen's inequality is small.\\
\indent Next, in Fig.~\ref{fig:SE_logN}, we examine the effect of the number of IRS elements, $N$, by plotting the ergodic sum-SE vs. $\log_2(N)$ for transmit SNRs of $130$ dB and $150$ dB to validate the scaling of the received SNR as $N^2$ for operator X and as $N$ for operator Y. On the plot we mark the slope of the different curves; and as expected from  Theorem~\ref{thm:rate_characterization}, it clear that while received SNR for a user served by operator X scales as $N^2$, it also scales as $N$ for a user served by operator~Y.\\
\begin{figure}[t]
	\vspace{-0.15cm}
	\centering
	\includegraphics[width=\linewidth]{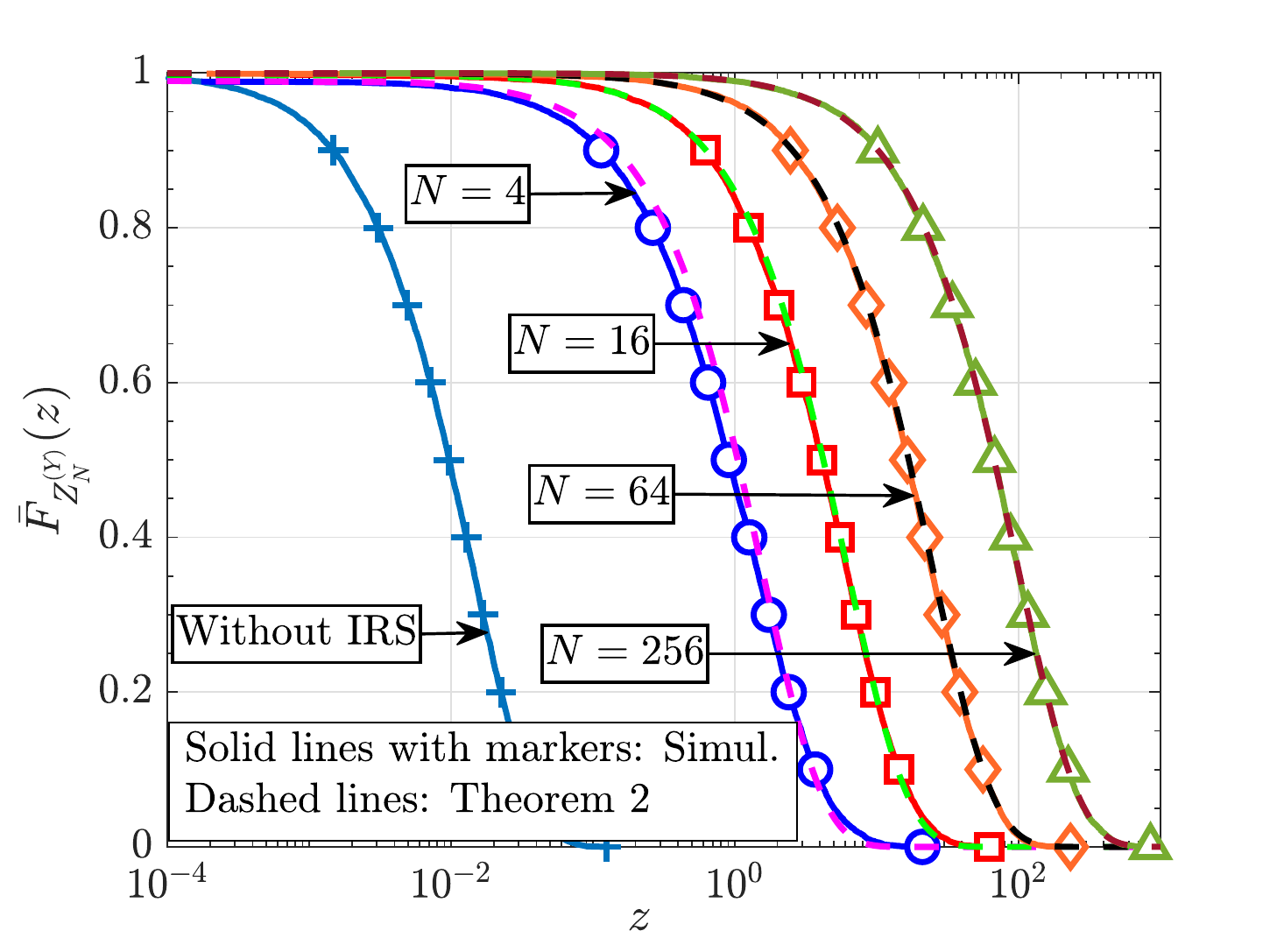}
	\caption{CCDF of $Z^{(Y)}_N$ as a function of $N$.} 
	\label{fig:ccdf}
	\vspace{-0.5cm}
\end{figure}
\indent Finally, we study the effect of the IRS on the OOB operator (namely, Y), by considering the behavior of the random variable $Z^{(Y)}_N$ (see~\eqref{eq:snr_offset}),
which represents the difference in the SNR/channel gain at a UE $q$ (OOB-UE) served by BS-Y with and without the IRS in the environment. In Fig.~\ref{fig:ccdf}, we plot the CCDF of $Z^{(Y)}_N$, given by~\eqref{eq:ccdf_notation}. 
\textcolor{black}{Firstly, we notice that the analytical expression as given in Theorem~\ref{thm:exact_ccdf} matches well with the simulations which validates the accuracy of Theorem~\ref{thm:exact_ccdf} even for smaller values of $N$.} Next, we observe that $Z^{(Y)}_N$ is a non-negative random variable for any $N>0$, which again confirms that almost surely, every possible outcome of the channel gain at an OOB UE \emph{with} an IRS is at least as good as every possible outcome of the channel gain at the same UE \emph{without} an IRS. Finally, we observe that the CCDF shifts to the right as the number of IRS elements is increased. On the same plot, we also show the CCDF of received SNR in the absence of IRS, which is in the left-most curve in the figure. This shows that the probability that an operator benefits from the presence of a randomly configured IRS in the vicinity  increases with $N$, even for operators who do not control the IRS. These observations  confirm our inference from \textcolor{black}{Proposition~\ref{prop:stochastic_dominance_sub_6}}. Further, the instantaneous SNRs witnessed at an arbitrary UE of an OOB operator stochastically dominates the SNR seen by the same UE in the absence of the IRS. Thus, the IRS only enhances the performance of any operator regardless of the frequency band of operation.
\vspace{-0.2cm}
\section{Conclusions}
\vspace{-0.15cm}
In this paper, we analyzed the effect of deploying an IRS on the  performance of an OOB operator that has no control over the IRS. We showed that while the IRS optimally serves the in-band UEs, it simultaneously, and at no additional cost, enhances the quality of the channels of the OOB UEs. This enhancement in the OOB case is a result of reception of multiple copies of the signal at the UEs. Our numerical results corroborate our theoretical results, and we conclude that deployment of an IRS benefits all the co-existing network operators, albeit to a lesser extent than the operator that has control over the IRS configurations. Future work can include the effect of multiple antennas at the BSs/UEs, and consider better scheduling schemes. In particular, with opportunistic scheduling, and if sufficiently many users are served by operator Y, any IRS configuration selected for UE served by BS-X can be near-optimal to some UE served by BS-Y also. Selecting and serving such an UE can procure near-optimal benefits from the IRS to both operators.
\vspace{-0.15cm}
\begin{appendices}
\renewcommand{\thesectiondis}[2]{\Alph{section}:}
\vspace{-0.15cm}
\section{Proof of Theorem~\ref{thm:rate_characterization}}\label{sec:app-1}
\vspace{-0.1cm}
	 \subsection{System ergodic sum-SE of operator X}	 
	  We first compute $\langle R_k^{(X)} \rangle$ for a given $k$. By Jensen's inequality, we obtain
	  \vspace{-0.1cm}
		\begin{equation}\label{eq:airtel_jensen}
		\vspace{-0.1cm}
			\langle R_k^{(X)} \rangle \leq	
			\log_{2}\!\left(\!1 \!+\! \mathbb{E}\left[\left|\sum\nolimits_{n=1}^N  |f^X_{n}{g_{k,n}}| 
				+   |h_{d,k}| \right|^2\right]\! \frac{P}{\sigma^2}\! \right)\!.
				\vspace{-0.1cm}
		\end{equation}
		We expand the expectation term as follows.
		\vspace{-0.1cm}
		\begin{multline}\label{eq:op_X_mean_expand}
		\vspace{-0.2cm}
			\left|\sum_{n=1}^N  |f^X_{n}{g_{k,n}}| 
			+   |h_{d,k}| \right|^2 = \sum_{n,m=1}^N |f^X_{n}||g_{k,n}||f^X_{m}||g_{k,m}| +  \\ |h_{d,k}|^2+ 2\!\left(\sum_{n=1}^N |f^X_{n}||g_{k,n}||h_{d,k}| \!\right) \!=\! |h_{d,k}|^2 + \!\sum_{n=1}^{N}|f^X_{n}|^2|g_{k,n}|^2  \\+
			\mathop{\sum_{n,m=1}^{N}}_{n \neq m} |f^X_{n}||g_{k,n}||f^X_{m}||g_{k,m}|  + 2\sum_{n=1}^N |f^X_{n}||g_{k,n}||h_{d,k}|.
			\vspace{-0.5cm}
			\end{multline}
Under Rayleigh fading, $\mathbb{E}[|f_n^X|^2] = \beta_{\mathbf{f}^X}, \mathbb{E}[|g_{k,n}|^2] = \beta_{\mathbf{g},k}, \mathbb{E}[|h_{d,k}|^2] = \beta_{d,k}$, $ \mathbb{E}[|f_n^X|] = \sqrt{\dfrac{\pi}{4}\beta_{\mathbf{f}^X}},$ $ \mathbb{E}[|g_{k,n}|] = \sqrt{\dfrac{\pi}{4}\beta_{\mathbf{g},k}}, \mathbb{E}[|h_{d,k}|]   = \sqrt{\dfrac{\pi}{4}\beta_{d,k}}, \hspace{0.1cm}\forall k \in [K], n \in [N]$. Further, all the random variables are independent. Taking the expectation in~\eqref{eq:op_X_mean_expand}, and substituting for these values, we get 
\vspace{-0.2cm}
		\begin{multline}\label{eq:expectation_iid_airtel}
		\vspace{-0.1cm}
			\mathbb{E}\left[\left|\sum\nolimits_{n=1}^N  |f^X_{n}{g_{k,n}}| 
			+   |h_{d,k}| \right|^2\right] = 
			N^2\left(\frac{\pi^2}{16} \beta_{r,k}\right) \\ + N\left(\beta_{r,k}-\frac{\pi^2}{16}\beta_{r,k}+\frac{\pi^{3/2}}{4}\sqrt{\beta_{d,k}\beta_{r,k}}\right) +\beta_{d,k}.
			\vspace{-0.1cm}
				\end{multline}
Substituting~\eqref{eq:expectation_iid_airtel} in~\eqref{eq:airtel_jensen}, and plugging in the resulting expression in~\eqref{eq:sum-rate-template} yields~\eqref{eq:rate_airtel_rr}, as desired.
\vspace{-0.4cm}	
	 \subsection{System ergodic sum-SE of operator Y}	 
As above, from Jensen's inequality, we have
\vspace{-0.1cm}
\begin{equation}\label{eq:jio_jensen}
\vspace{-0.2cm}
\langle R_q^{(Y)} \rangle \leq \log_{2}\left(\!1 + \mathbb{E}\left[\left|\sum\nolimits_{n=1}^N  f^Y_{n}{g_{q,n}} 
				+   h_{d,q} \right|^2\right] \frac{P}{\sigma^2} \right)\!.
			\end{equation}
		Proceeding along the same lines and simplifying, we get
			\vspace{-0.1cm}		
		\begin{equation}\label{eq:expectation_iid_jio}
		\vspace{-0.1cm}
			\mathbb{E}\left[\left|\sum\nolimits_{n=1}^N  f^Y_{n}{g_{q,n}} 
			+   h_{d,q} \right|^2\right] = N\beta_{r,q} + \beta_{d,k}.
		\vspace{-0.1cm}
		\end{equation}
		Substituting~\eqref{eq:expectation_iid_jio}  in~\eqref{eq:jio_jensen}, and plugging in the resulting expression in~\eqref{eq:sum-rate-template} yields~\eqref{eq:rate_jio_rr}.	 \qed
		\vspace{-0.36cm}		
		\section{Proof of Theorem~\ref{thm:exact_ccdf}}\label{sec:ccdf_proof}
		For large $N$, we have $h_{1,q} \sim \mathcal{CN}(0,N\beta_{r,q}+\beta_{d,q})$~\cite[Proposition~$1$]{Yashvanth_TSP_2023}.\footnote{We consider large $N$ for the sake of analytical tractability. We
		showed in~\cite{Yashvanth_TSP_2023} that this approximation works well even with $N = 8$.} So, $|{h}_{1,q}|^2 \sim \exp(1/\left(N\beta_{r,q}+\beta_{d,q}\right))$, and $|{h}_{2,q}|^2\sim \exp(1/\beta_{d,q})$. Define the real-valued random variables $\tilde{h}_{1,q} \triangleq |h_{1,q}|^2$, and $\tilde{h}_{2,q} \triangleq |h_{2,q}|^2$. Hence, $Z^{(Y)}_N$ is the difference of two non-identically and exponentially distributed random variables. We first compute the correlation coefficient between $\tilde{h}_{1,q}$ and $\tilde{h}_{2,q}$. 
The correlation coefficient is given by
\vspace{-0.2cm}
			\begin{equation}\label{eq:corr_coeff}
				\vspace{-0.2cm}
		\rho_{12} \triangleq \left(\mathbb{E}\left[(\tilde{h}_{1,q}-\mathbb{E}[\tilde{h}_{1,q}])(\tilde{h}_{2,q}-\mathbb{E}[\tilde{h}_{2,q}])\right]\right)/ \sigma_1\sigma_2,
\end{equation} where $\sigma_1^2$ and $\sigma_2^2$ are the variances of $\tilde{h}_{1,q}$ and $\tilde{h}_{2,q}$, respectively. We also note $\mu_1 \triangleq \mathbb{E}[\tilde{h}_{1,q}] = N\beta_{r,q} + \beta_{d,q}$, $\mu_2 \triangleq \mathbb{E}[\tilde{h}_{2,q}] =  \beta_{d,q}$, $\sigma_1^2 = {\left(N\beta_{r,q} + \beta_{d,q}\right)}^2$, $\sigma_2^2 =  \beta_{d,q}^2$.
 Thus, from~\eqref{eq:corr_coeff}, we get
 \vspace{-0.1cm}
\begin{equation}
	\vspace{-0.1cm}
	\rho_{12} = \frac{\mathbb{E}[\tilde{h}_{1,q}\tilde{h}_{2,q}] - \left(N\beta_{r,q} + \beta_{d,q}\right)\beta_{d,q} }{{\left(N\beta_{r,q} + \beta_{d,q}\right)}\beta_{d,q}}.
\end{equation} 
Using the  expressions for $\tilde{h}_{1,q}$ and $\tilde{h}_{2,q}$ from~\eqref{eq:aux_RV_defn}, we can verify that $\mathbb{E}[\tilde{h}_{1,q}\tilde{h}_{2,q}] = N\beta_{r,q}\beta_{d,q} + 2\beta_{d,q}^2$, and obtain
\vspace{-0.1cm}
\begin{equation}\label{eq:rho_value}
	\vspace{-0.1cm}
	\rho_{12} = 1\left.\middle/\right. \left(1 + N\left(\beta_{r,q}/\beta_{d,q}\right) \right).
	\vspace{-0.1cm}
\end{equation} 
Clearly, $\rho_{12}$ decays inversely with $N$. We now use a result from~\cite{Simon_2002_ProbabilityGaussian} which characterizes the distribution of the difference of two dependent and non-identically distributed chi-square random variables and obtain the CDF of $Z^{(Y)}_N$ as
\vspace{-0.2cm}
\begin{equation}\label{eq:cdf_simon}
 	F_{Z^{(Y)}_N}(z) = \left\{
 	\begin{array}{lr}
	 		\!\dfrac{8}{\mu_1\mu_2(1\!-\rho_{12}^2)\gamma\alpha^-}e^{\left(\frac{\alpha^- z}{4}\right)}, & \text{if } z < 0,\\ 
	 		1 - \dfrac{8}{\mu_1\mu_2(1-\rho_{12}^2)\gamma\alpha^+}e^{-\left(\frac{\alpha^+ z}{4}\right)}, & \text{if } z\geq 0.
	 	\end{array}
 	\right.
 	\vspace{-0.4cm}
\end{equation} 
where 
\begin{equation}\label{eq:cdf_simon_parameters}
\!	\!\gamma\! = \!\dfrac{\!2\sqrt{{\!(\mu_2\! -\!\mu_1)}^2\!\!+\!4\mu_1\mu_2(1\!-\!\rho_{12}^2)}}{\mu_1\mu_2(1\!-\rho_{12}^2)\!},	\alpha^\pm\!=\! \gamma \!\pm\! \dfrac{2\left(\mu_2\!-\!\mu_1\right)}{\mu_1\mu_2(1\!-\!\rho_{12}^2)},
\end{equation} We simplify the CDF by considering a large $N$; and let $\rho_{12} \rightarrow 0$ in~\eqref{eq:cdf_simon}, \eqref{eq:cdf_simon_parameters} as per~\eqref{eq:rho_value}. Finally, the CCDF of $Z^{(Y)}_N$, $\bar{F}_{Z^{(Y)}_N}(z) = 1 - F_{Z^{(Y)}_N}(z)$ is obtained  as
\vspace{-0.2cm}
\begin{equation} \label{eq:CCDF_Simon1}
	\vspace{-0.1cm}
	\bar{F}_{Z^{(Y)}_N}(z) = \left\{
	\begin{array}{lr}
		1-\dfrac{\mu_2}{\mu_1+\mu_2}e^{\frac{z}{ \mu_2}}, & \text{if } z < 0, \\ 
		 \dfrac{\mu_1}{\mu_1+\mu_2}e^{-\frac{z}{ \mu_1}}, & \text{if } z\geq 0,
	\end{array}
	\right.
	\vspace{-0.05cm}
\end{equation} 
Substituting for $\mu_1$ and $\mu_2$ into \eqref{eq:CCDF_Simon1} completes the proof. \qed
\vspace{-0.3cm}
\end{appendices}
\vspace{-0.3cm}
\bibliographystyle{IEEEtran}
\bibliography{IEEEabrv,IRS_ref_short}
\vspace{-0.1cm}
\end{document}